\documentclass[letterpaper, 11pt]{article}

\title{Towards a Lower Bound for the Average Case Runtime of Simulated Annealing on TSP}

\usepackage{algorithm2e}
\usepackage{amsmath}
\usepackage{amssymb}
\usepackage{amsthm}
\usepackage{thmtools}
\usepackage{thm-restate}
\usepackage{cleveref}
\usepackage{graphicx}
\usepackage[letterpaper, margin=1in]{geometry}

\usepackage{authblk}

\author[1]{Bodo Manthey}
\author[1]{Jesse van Rhijn}
\affil[1]{Faculty of Electrical Engineering, Mathematics, and Computer Science,
University of Twente}

\usepackage[style=numeric]{biblatex}
\AtEveryBibitem{
    \clearfield{urlyear}
    \clearfield{urlmonth}
}

\let\given\givenbase

\newcommand{\diag}{\operatorname{diag}}

\def\prob{\ensuremath\mathbb{P}}
\def\expect{\ensuremath\mathbb{E}}

\theoremstyle{definition}

\theoremstyle{remark}
\newtheorem*{remark}{Remark}

\theoremstyle{plain}
\newtheorem*{claim}{Claim}
\declaretheorem[name=Theorem,numberwithin=section]{thm}
\newtheorem{corollary}{Corollary}[thm]
\newtheorem{lemma}[thm]{Lemma}

\newcommand*{\myproofname}{Proof}
\newenvironment{claimproof}[1][\myproofname]{\begin{proof}[#1]}{\end{proof}}

\begin{document}

\maketitle

\begin{abstract}
  We analyze simulated annealing (SA) for simple randomized instances of the 
    Traveling Salesperson Problem.
    Our analysis shows that the theoretically optimal
    cooling schedule of Hajek \cite{hajek_cooling_1988}
    explores members of
    the solution set which are in expectation far from the global optimum.
    We obtain a lower bound on the expected length of the
    final tour obtained by SA on these random instances.
    In addition, we also obtain an upper bound on the expected value of its
    variance. 
    These bounds assume that the Markov chain that describes SA is
    stationary, a situation that does not truly hold in practice. Hence, we also
    formulate conditions under which the bounds extend to
    the nonstationary case.
    These bounds are obtained
    by comparing the tour length distribution to a related distribution.
    We furthermore provide numerical
    evidence for a stochastic dominance relation that appears to exist
    between these two distributions, and formulate a conjecture
    in this direction. If proved, this conjecture implies
    that SA stays far from the global optimum with high probability when
    executed for any
    sub-exponential number of iterations. This would show that
    SA requires at least exponentially many iterations
    to reach a global optimum with nonvanishing probability.
\end{abstract}

\section{Introduction}

Discrete minimization problems require one to minimize an objective function over a
finite set of solutions. In many cases of interest, the computational effort
required to find a globally optimal solution is assumed to be prohibitively large.
A classical example
is the Traveling Salesperson Problem (TSP): given a weighted graph $G$,
the goal is to find a minimum-weight Hamiltonian cycle
of $G$. Although easily stated, the TSP is NP-hard
\cite{korte_combinatorial_2000}.

To nevertheless
obtain solutions that are satisfactory, a common resolution is to apply
local search heuristics rather than exact algorithms. Although these heuristics
are not guaranteed to converge to the global optimum, they often find solutions
of surprisingly high quality. Moreover, the problem of finding local optima
appears to be tractable on practical instances, despite the pessimistic
worst case performance of many of these heuristics
\cite{vattani_k-means_2011, arthur_smoothed_2011, englert_worst_2014, chandra_new_1999}.
In particular, the $k$-opt heuristic for TSP gives near-optimal
solutions for Euclidean instances \cite{cohen-addad_effectiveness_2015},
and runs in polynomial time in practice, as demonstrated
in a study by Johnson \& McGeoch
\cite[chapter 8]{aarts_local_2003}.

While local improvement heuristics are quite successful, they are prone to getting
stuck in local optima. An approach to solving this problem can be found in
metaheuristics such as simulated annealing (SA) 
\cite{aarts_simulated_1989,geman_stochastic_1984}. At the basis of SA lies a local
improvement heuristic. However, rather than requiring every step of the optimization process
to strictly decrease the objective function, a probabilistic approach is taken wherein
the objective function is allowed to increase. As such, SA can
escape a local optimum in which a local improvement heuristic would get stuck.

To control the probability with which these increasing steps are taken,
a control parameter called the temperature is used. At low temperatures, the
probability of increasing the objective function should become small; indeed, SA is effectively identical to its underlying local search
heuristic at a temperature of 0.
Previous results \cite{faigle_note_1991,hajek_cooling_1988} show that
for sufficiently slowly decreasing temperature, simulated annealing
asymptotically converges to the global optimum. Even at a finite
temperature, one can show that the global optimum is eventually reached with high probability
\cite{rajasekaran_convergence_1990}.

Unfortunately, while rigorous, these bounds are not practically useful. In the case
of the TSP, they can only guarantee convergence within $O(n^n)$ steps
\cite{aarts_simulated_1989, rajasekaran_convergence_1990}, which is more
costly than even a complete enumeration of the search space.
In practice, the algorithm
is thus terminated long before convergence is guaranteed.

These upper bounds are furthermore taken with respect to all instances,
and thus can be considered worst-case bounds. As the typical runtime of
local search heuristics is often much better than the worst case runtime,
a natural question to ask is then if we can prove a better
bound for simulated annealing on average-case instances. 

Here, we set out towards a rigorous lower bound on the expected time to
reach the optimal TSP tour for a simple class of
random TSP instances. 
These results are formally stated in \Cref{sec:results}.
We prove a lower bound on the expected value of
the tour length obtained by the algorithm if the underlying
Markov chain is assumed to be in equilibrium. Furthermore, we
show that this bound also holds out of equilibrium, assuming the
Markov chain is sufficiently mixed. Finally, we formulate a conjecture on the
statistical behavior of the tour length. Proving this conjecture would allow us
to obtain a rigorous lower bound on the runtime of simulated annealing
for the logarithmic cooling schedule derived by Hajek \cite{hajek_cooling_1988}.

\section{Preliminaries}

\subsection{Notation and Definitions}

Throughout the paper, we use $\mathbb{R}_+$ to refer
to the nonnegative reals, and $\mathbb{R}_{++}$ for the positive reals.

Consider any discrete minimization problem with finite solution space $S$ and
associated objective function $J : S \to \mathbb{R}_+$. We denote by $n \in \mathbb{N}$
the \emph{problem size}, which
is simply the number of nodes of the underlying TSP instance in most cases in this paper.

Since simulated annealing defines a Markov chain on $S$, it is convenient to consider
the state graph of this Markov chain. To that end, we first define a neighborhood structure
on $S$, which is simply a function $N : S \to 2^S$. Although the neighborhood function
can be arbitrary in principle, it is usually highly structured. The state
graph is then given by the directed graph $H = (S, A)$, where we have
$(x, y) \in A$ iff $y \in N(x)$.

The neighborhood function is often symmetric, in the sense that if
$x \in N(y)$ then $y \in N(x)$. In this case, we will regard the state graph as undirected.
This definition of the neighborhood differs from the one usually used
in the context of local search, where $x \in N(y)$ only if $J(x) < J(y)$. We
relax this requirement here, since simulated annealing may increase the objective
function in any given iteration.

In this paper, we will always assume that the neighborhood $N$ is symmetric. In other words,
we consider only local search heuristics with reversible iterations. Hence, the state
graph $H$ considered here is undirected. In addition, we assume throughout that $H$
is $d$-regular, where $d$ may be a function of $n$. In particular, for the TSP
with the 2-opt heuristic, the state graph has $|S| = (n-1)!/2$ and is
$n(n-3)/2$-regular.

We call a nonincreasing function $T : \mathbb{N} \to \mathbb{R}_{++}$
a cooling schedule. In the context of the physical process that SA takes its
inspiration from, a cooling schedule determines the rate at which the temperature
of a material is decreased with time.
The choice of the cooling schedule is in principle arbitrary. A trade-off exists between the
quality of the solution obtained, and the rate at which $T$ decreases. Slowly
decreasing cooling schedules potentially yield optimal solutions \cite{hajek_cooling_1988},
but may be impractical \cite{aarts_simulated_1989}.

Since SA uses a Markov chain, it is a randomized algorithm. In this work, we will analyze
the performance of SA on random \emph{instances}. Thus, the involved Markov chains
and their stationary distributions will be functions of the randomness present in the
instances. In order to meaningfully assign probabilities, the formalism of
compound distributions is useful.

Suppose $\pi(x \given \theta)$ denotes a probability measure on $S$, dependent on a set
of parameters $\theta$. We write $\expect_\pi(f \given \theta)$ for the conditional
expected value of $f : S \to \mathbb{R}$. If $\theta$ are themselves distributed according
to $\mu$, we write the unconditional or \emph{compound} expectation
$\expect(f) = \expect_\mu(\expect_\pi(f \given \theta))$. Moreover, we have the
joint 
distribution $p(x, \theta) = \pi(x\given \theta)\mu(\theta)$, which is the
generalization of the mixture distribution \cite{mood_introduction_1973} to the
uncountable setting.

Finally, to avoid ambiguities, we provide a definition of the version of simulated annealing
we consider in this paper in Algorithm \ref{sa_def}. 

\begin{algorithm}[t]
    \KwData{Solution space $S$, neighborhood function $N$, cooling schedule $T$, a fixed
    number of iterations $t \in \mathbb{N}$, and a state $x_0 \in S$.}
    \For{$i \in [t]$}{
        generate a state $y \in N(x_i)$ uniformly at random\;
        $\Delta \gets J(y) - J(x_i)$\;
        \eIf{$\Delta < 0$}{
            $x_{i+1} \gets y$\;
        }{
            $p \gets \exp\left(-\Delta/T(i)\right)$\;
            set $x_{i+1}$ to  $y$ with probability $p$ and
            to  $x_i$ with probability $1-p$\;
        }
    }
    \textbf{return} $x_t$
\caption{Simulated Annealing \cite{rajasekaran_convergence_1990}.\label{sa_def}}
\end{algorithm}

\subsection{Random Instance Model}

We next define the model we use to generate random TSP instances,
and the relevant probability distributions. The model we use
is a symmetric version of a model introduced by Karp \cite{karp_patching_1979, el_lawler_chapter_1985}.

Let $G = (V, E)$ be a complete graph with
$|V| = n$ and $|E| = n(n-1)/2 = m$,
and define the edge weight function $w : E \to [0, 1]$. We consider a model where each edge
weight is uniformly 
and independently
drawn from the unit interval: $w(e) \sim U[0, 1]$ for each $e \in E$. When convenient
we will consider the weights as a uniform random vector in $[0, 1]^m$, indexed
by $E$. With a slight abuse
of notation, we use the same symbol $w$ for this vector. The
distribution of this weight vector will be denoted by $\mu$ throughout.

At fixed temperature, Algorithm \ref{sa_def} defines a Markov chain with transition
matrix
\[
    \tilde{P}_T(x, y \given w) = \begin{cases}
        0, & y \notin N(x) \cup \{x\},  \\
        \frac{1}{d}\min\left\{1, \exp\left(-\frac{J(y \given w) - J(x \given w)}{T}\right)\right\}, & y \in N(x), \\
        1 - \sum_{z \in N(x)} \tilde{P}_T(x, z \given w), & y = x,
    \end{cases}
\]
where $J(x \given w) = \sum_{e \in x} w(e)$. We use the conditional notation here
to emphasize that $J$ is a random variable, dependent on the weights $w$.
If the Markov chain defined by simulated annealing is in equilibrium at a certain temperature
$T$, then the probability that the chain is in state $x \in S$ at any time can be computed
for a fixed set of edge weights $w \in [0, 1]^m$ as follows: 
\[
    \tilde{\pi}_T(x \given w) = \frac{1}{\tilde{Z}(T \given w)}e^{-J(x \given w)/T},
\]
where
\[
    \tilde{Z}(T \given w) = \sum_{x \in S} e^{-J(x \given w)/T}
\]
is called the \emph{partition function}, which
serves to ensure that $\tilde{\pi}_T$ is a probability measure. 
It will prove more convenient to work with the inverse
temperature $\beta = 1/T$. Therefore, we define $Z(\beta \given w) = \tilde{Z}(1/\beta \given w)$. Likewise, we set
$\pi_\beta(x \given w) = \tilde{\pi}_{1/\beta}(x \given w)$ and $P_\beta = \tilde{P}_{1/\beta}$.

Due to the symmetry of our random instance model, it can be convenient
to recast the stationary distribution in terms of the tour length, rather
than in terms of $x$ and $W$.
To do so, we first define
\[
    Q(W) = \{J(x \given W) \given x \in S\},
\]
the set of tour lengths for a given realization of the edge weights, $W$.
This
allows us to formally define
\[
    \rho(j \given W) = \frac{1}{|S|} \sum_{j' \in Q(W)} d(j' \given W)\delta(j - j'),
\]
which is the density of tours of length $j$, given $W$. Here,
$d(j\given W)$ denotes the number of tours
with length exactly $j$ for a given realization of the edge weights,
and $\delta(\cdot)$ denotes the Dirac delta.
We also define
\[
    \rho(j) = \int_{[0,1]^m} \rho(j\given W=w)\mu(w) \mathrm{d}w.
\]
If $\mu$ is the uniform distribution on $U[0,1]^m$, as we assume here,
then $\rho$ is exactly the distribution of the sum of $n$ independent
$U[0,1]$ variables. To see this, note that sampling from $\rho$
is equivalent to sampling a tour $x$ from $S$ uniformly, and computing its
length given the random edge weights. Since each edge of $x$ is $U[0,1]$,
it follows that its length is distributed as claimed.

With these definitions, we can now write a formal expression for the
distribution of $J$, at inverse temperature $\beta$:
\begin{align}\label{eq:quenched}
    \pi^q_\beta(j) = \int_{[0,1]^m}\left(\frac{e^{-\beta j} \rho(j \given W=w)}
        {\int_0^n e^{-\beta j'}\rho(j'\given W=w)\mathrm{d}j'}\right) \mu(w)\mathrm{d}w.
\end{align}
The meaning of the superscipt $q$ will be clarified in a moment. First, let
us also define a related distribution,
\begin{align}\label{eq:annealed}
    \pi^a_\beta(j) = \frac{e^{-\beta j} \rho(j)}{\int_0^n e^{-\beta j'}\rho(j')\mathrm{d}j'}.
\end{align}

These quantities can be given an interpretation,
which is well-known in the physics of disordered systems \cite{menon_physics_2012}.
The distribution $\pi_\beta^q$ is called the \emph{quenched} distribution. 
It corresponds to the situation where one first draws the random edge variables
from $\mu$, and subsequently performs Algorithm \ref{sa_def} at constant temperature $\beta^{-1}$,
until the Markov chain is in equilibrium. Then
$\pi^q_\beta$ is the distribution of the tour length obtained
at the end of this process.

In contrast, $\pi^a_\beta$ is called
the \emph{annealed} distribution. In this situation, one allows the edge lengths
to change dynamically during execution of the algorithm, defining
a Markov chain on the extended state space
$S \times [0, 1]^m$.
Suppose at some iteration,
the Markov chain is in the state $(x, w) \in S \times [0,1]^m$. One
then randomly chooses an edge, and either increases or decreases its
weight, obtaining a new weight vector $\bar{w}$.
One then selects a new tour $y \in N(x)$. The acceptance step
in Algorithm \ref{sa_def} is then performed, comparing $J(x\given w)$ to $J(y\given \bar{w})$.

Clearly, the quenched setting is the one relevant to our average case model;
however, its statistics are quite complicated.
Indeed,
so-called quenched disorder is an active area of research in mathematics, physics and
computer science; for a fully rigorous work in the field, see e.g.\ Talagrand 
\cite{talagrand_spin_2003}. Meanwhile, the statistics of the annealed
distribution are much simpler. Luckily, the annealed distribution still
provides some information that we can use.

\subsection{Summary of Main Results and a Conjecture}\label{sec:results}

In \Cref{sec:equilibrium}, we obtain the following theorem and corollary.

\begin{restatable*}{thm}{equilibriumcost}\label{lemma:cost}
      Let $J \sim \pi_\beta^q$, cf. \Cref{eq:quenched}, and
      let $\sigma^2(W) := \mathrm{Var}_{\pi_\beta}(J \given W)$. Then
      \begin{align*}
          \expect(J) \geq n\left(
                \frac{1}{\beta} - \frac{1}{e^\beta - 1} 
          \right)
            \quad \text{and} \quad 
            \expect(\sigma^2(W)) \leq
                \frac{n}{\beta^2}.
      \end{align*}
\end{restatable*}

\begin{restatable*}{corollary}{equilibriumcostcorollary}\label{corollary:cost}
    Assuming simulated annealing is stationary at iteration $t$, 
    the logarithmic cooling schedule $T(t) = a/\ln(t)$ with
    $a > 0$ yields
    \[
        \expect(J(t)) \geq n\left(
            \frac{a}{\ln t} - \frac{1}{t^{1/a} - 1}
        \right)
        \quad \text{and} \quad
        \expect(\sigma^2(W)) \leq \frac{a^2 n}{\ln^2 t}.
    \]
\end{restatable*}

\Cref{lemma:cost} tells us that for constant $\beta$, the expected tour length
returned by Algorithm \Cref{sa_def} is $\Omega(n)$. Contrasting this
with the fact that the optimal tour is $\Theta (1)$ with high probability
\cite{frieze_random_2004}, we see that simulated annealing
stays far from the optimal tour when the temperature is kept constant.

Extending this result beyond equilibrium, which
we do in \Cref{sec:mixing}, leads to the following theorem.
For the purposes of this theorem, we consider an \emph{epoch}
of Algorithm \ref{sa_def} as a sequence of subsequent iterations in which
the temperature is kept constant. 

\begin{restatable*}{thm}{nonequilibriumcost}\label{thm:lb_nonequilibrium}
    Suppose Algorithm \ref{sa_def}
    is executed for $t = 2^{o(n)}$ epochs with a 2-opt neighborhood
    and with the logarithmic cooling schedule $T(t) = a/\ln(t)$, taking $a \geq n$.
    Assume that at epoch $k$,
    the temperature is kept constant for 
    $\Omega(n^9 k^2 \log n)$ iterations.
    Then the lower bound on $\expect(J(t))$ from
    \Cref{corollary:cost} holds.
\end{restatable*}

Finally, in \Cref{sec:tailbound}, we formulate a conjecture relating the
annealed and quenched distributions. 

\begin{restatable*}{conjecture}{stochdomconj}\label{conjecture:stochdom}
    Let $J_a \sim \pi_\beta^a$
    and $J_q \sim \pi_\beta^q$.
    For any $n \geq 3$, we have $J_q \succeq J_a$.
\end{restatable*}

We also prove this conjecture for $n=3$, but have been unable to extend
it to $n > 3$. Hence, we instead provide some numerical evidence for
\Cref{conjecture:stochdom} in this section.

\section{Tour Length in Equilibrium}\label{sec:equilibrium}

In this section we prove a lower bound on the length of a tour in the TSP after running
Algorithm \ref{sa_def} for a certain number of steps $t$, using the logarithmic cooling
schedule
\begin{align}\label{eq:logcooling}
    T(t) = \frac{a}{\ln(t)}.
\end{align}
Here, $a = a(n)$
is a positive real number that may depend on $n$. For simplicity of notation,
we often omit the dependence on $n$, and prefer to explicitly state
when it is relevant. The origin of this cooling schedule
is the work of Hajek \cite{hajek_cooling_1988}. Hajek showed that this schedule
guarantees that $\pi_\beta$ will concentrate on the global minima of the instance as
$t \to \infty$, provided
that $a$ satisfies a condition related to the distribution of tour lengths. For
our random instance model, $a \geq n$ suffices. Some of our results 
(e.g.\ \Cref{lemma:moments_nonequilibrium,thm:lb_nonequilibrium}) indeed carry this
restriction on $a$; whenever this is the case, it will be explicitly stated.

We first formalize the fact that
the partition function contains useful information concerning the distribution
$\pi_\beta$, as it is related to the cumulant generating function
associated with this distribution. For the sake of clarity, we provide a brief definition of
the cumulant generating function here; for a more thorough treatment, see Kendall
\cite{alan_stuart_kendalls_2010}.

Given a random variable $X$ with moment generating function
$M_X$, its cumulant generating function is defined as
\[
    K_X(s) := \ln M_X(s) = \ln \expect(e^{sX}).
\]
The $p^\text{th}$ cumulant of $X$ is then given by
\[
    \kappa_p(X) = \left.\frac{\mathrm{d}^p}{\mathrm{d}s^p} K_X(s)\right|_{s=0}.
\]
Note the connection between the generating functions for the cumulants
and moments of $X$. Cumulants and moments can be expressed
as polynomials of one another. For the first two, the relationship is very straightforward:
$\kappa_1(X) = \expect(X)$, and $\kappa_2(X) = \mathrm{Var}(X)$.

\begin{restatable}{lemma}{cumulants}\label{lemma:cumulants}

    Let $\kappa_p(\beta \given W)$ denote the $p^\mathrm{th}$ cumulant of the cost function with respect to the
    distribution $\pi_\beta(\cdot \given W)$. We have
    \[
        \kappa_p(\beta \given W) = (-1)^p \frac{\partial^p}{\partial\beta^p}
            \ln Z(\beta \given W). 
    \]
\end{restatable}

\begin{corollary}\label{cor:mean_variance}
    Let $W$ be an arbitrary $[0,1]^m$-valued random variable. Then
    \begin{align}\label{eq:mean}
        \expect_{\pi_\beta} (J \given W)
             = -\frac{\partial}{\partial \beta} \ln Z(\beta \given W) \quad \text{and} \quad
        \mathrm{Var}_{\pi_\beta}(J \given W) = \frac{\partial^2}{\partial\beta^2} \ln Z(\beta \given W).
    \end{align}
\end{corollary}

Notice that $\expect_{\pi_\beta}(J \given W)$
and $\mathrm{Var}_{\pi_\beta}(J \given W)$ are themselves random variables with
respect to the edge weights. 
Using the following series of lemmas, we show a lower bound for the expected value of the objective
function and an upper bound for its variance (\Cref{lemma:cost}).
For both bounds, the expected value is taken over both the random edge weights and the Markov chain.
The following is a direct result of the Leibniz rule:

\begin{lemma}\label{lemma:derivative_mean}
    \[
        \expect\left(\frac{\partial^k}{\partial\beta^k}\ln Z(\beta \given W)\right) = \frac{\mathrm{d}^k}{\mathrm{d}\beta^k}\expect(\ln Z(\beta \given W)). 
    \]
\end{lemma}

\Cref{lemma:derivative_mean} tells us that, at constant $\beta$,
\[
    \expect(J) = \expect_\mu \left( \expect_{\pi_\beta} (J \given W)\right)
    = -\frac{\mathrm{d}}{\mathrm{d}\beta}\expect(\ln Z(\beta \given W)).
\]
This identity allows us to find the expected length of a tour drawn from $\pi_\beta(\cdot \given W)$,
where the expectation is taken over $W \sim U[0, 1]^m$. One technical difficulty is that the expectation
of $\ln Z(\beta \given W)$ is not straightforward to compute. If it were possible to move the expectation inside the
logarithm, this would be of no issue, but this is not valid in general. Luckily,
we can show through \Cref{lemma:j_lower_bound} that $\ln \expect (Z(\beta \given W))$
still provides some information
on $\expect(J)$.

\begin{restatable}{lemma}{jlowerbound}\label{lemma:j_lower_bound}
    \[
        \frac{\mathrm{d}}{\mathrm{d}\beta} \expect(\ln Z(\beta \given W)) 
            \leq \frac{\mathrm{d}}{\mathrm{d}\beta} \ln \expect(Z(\beta \given W))
            \quad \text{and} \quad
        \frac{\mathrm{d}^2}{\mathrm{d}\beta^2} \expect(\ln Z(\beta \given W)) 
            \leq \frac{\mathrm{d}^2}{\mathrm{d}\beta^2} \ln \expect(Z(\beta \given W)).
    \]
\end{restatable}

\begin{restatable}{lemma}{partitionfunction}\label{lemma:partition_function}
    The expected value of the partition function for the average case model
    is given by
    \[
        \expect(Z(\beta \given W)) = \frac{(n-1)!}{2} \left(\frac{1 - e^{-\beta}}{\beta}\right)^n.
    \]
\end{restatable}

\equilibriumcost

\begin{proof}
    From \Cref{lemma:partition_function}, we have 
    \[
        \ln\expect(Z(\beta \given W)) = n\ln\left(\frac{1 - e^{-\beta}}{\beta}\right)
            + \Theta(n\log n).
    \]
    Differentiating to $\beta$,
    \[
        \frac{\mathrm{d}}{\mathrm{d}\beta}\ln \expect(Z(\beta \given W))
            = n\left(\frac{1}{e^\beta - 1} - \frac{1}{\beta}\right).
    \]
    Hence, 
    \begin{align}\label{eq:cost}
       -\frac{\mathrm{d}}{\mathrm{d}\beta} \ln \expect(Z(\beta\given W))
            = n\left(\frac{1}{\beta} - \frac{1}{e^\beta - 1}\right).
    \end{align}
    The result for $\expect(J)$ now follows from \Cref{lemma:j_lower_bound},
    together with \Cref{cor:mean_variance}. The result for $\expect(\sigma^2(W))$
    is similarly obtained; to wit,
    \begin{align}\label{eq:variance}
        \expect(\sigma^2(W)) \leq
            n\left(\frac{1}{\beta^2} - \frac{1}{e^\beta - 1}
                - \frac{1}{(e^\beta-1)^2}\right) \leq \frac{n}{\beta^2}.
    \end{align}
\end{proof}

In light of the cooling schedule of \Cref{eq:logcooling}, we define
$J(t)$ as the tour length drawn from $\pi_\beta^q$ with parameter
$\beta = T(t)^{-1}$.
Then \Cref{lemma:cost} yields the following corollary.

\equilibriumcostcorollary

\begin{remark}
    Aarts and Korst \cite[Postulate 2.1]{aarts_simulated_1989} proposed a distribution for the
    values of the cost function in a typical combinatorial optimization problem. Using this
    postulate, they derived asymptotic behavior for the expected value of the cost function 
    and its variance for small and large  temperatures. In particular, they concluded the
    following.
    \begin{itemize}
        \item For high temperatures, the expected cost is linear in $T^{-1}$ and
            the variance is constant.
        \item For low temperatures, the expected cost is linear in $T$
            and the variance is proportional to $T^2$.
    \end{itemize}
    Here, the problem size $n$ is taken as a constant, so that the asymptotic behavior is
    understood to be with respect to $T = 1/\beta$.
    With the results obtained in the proof of \Cref{lemma:cost}, we find that the heuristic results
    obtained
    by Aarts and Korst~\cite{aarts_simulated_1989}
    agree with our lower and upper bounds.
    
    Using \Cref{eq:cost}, for small $T$ we obtain
    $\expect(J) = \Omega(n/\beta) = \Omega(nT)$,
    so the cost is indeed lower bounded by a function linear in $T$ in this regime.
    For large
    $T$, we perform a Maclaurin expansion of
    \Cref{eq:cost} up to first order in $\beta$ to find
    \[
        \expect(J) \geq 
        n \left(\frac{1}{\beta} - \frac{1}{e^{\beta} -1}\right)
        = n\left(
                \frac{1}{2} - \frac{\beta}{12} + O(\beta^2) 
        \right) = n\left(\frac{1}{2} - \frac{1}{12T} + O\left(\frac{1}{T^2}\right)\right).
    \]
    We conclude that in this regime, the expected cost is lower-bounded by a function
    asymptotically linear in $T^{-1}$.
    
    As for the variance, we have from \Cref{lemma:cost} that $\expect(\sigma^2(W) = O(nT^2)$
    for small $T$. For large $T$, we now perform a Maclaurin expansion of
    \Cref{eq:variance} to find
    \[
        \expect(\sigma^2(W)) \leq n\left(\frac{1}{12} + O\left(\frac{1}{T^2}\right)\right),
    \]
    showing that for large $T$, the variance is bounded by a constant.
\end{remark}

\section{Tour Length Outside Equilibrium}\label{sec:mixing}

The results obtained in the previous section hold when the Markov chain that models simulated
annealing is in equilibrium. As this does not hold exactly for
any finite
$t > 0$, we now turn
to a non-equilibrium analysis. A central concept in this field is the
\emph{mixing time} of a Markov chain. Here we simply state the required definitions; for a comprehensive
overview, see Levin \& Peres \cite{aldous_markov_2019}. 

Let $P$ denote the transition matrix of any Markov chain with state space $S$. 
We denote the stationary distribution of this matrix by
$\pi$. Then we define the function
\[
    d(t) = \sup_{\nu \in \mathcal{P}} \|\nu P^t - \pi\|_\mathrm{TV},
\]
where $\|\cdot\|_{\mathrm{TV}}$ denotes the total variation distance \cite{aldous_markov_2019}
and $\mathcal{P}$ is the set of all probability measures on $S$.

In terms of $d(t)$, the mixing time of the Markov chain defined by $P^t$ is
\[
    \tau(\epsilon) := \min\{ t \given d(t) \leq \epsilon \}.
\]
Intuitively, the mixing time quantifies the number of applications of $P$ sufficient to
bring any probability measure $\epsilon$-close to $\pi$ in the total variation sense.

We are now ready to state a bound on the mixing time of the Markov chain defined
by simulated annealing.

\begin{restatable}{lemma}{mixingtime}\label{lemma:mixing_time}
    Assume that the state graph of the Markov chain defined by
    Algorithm \ref{sa_def} is $d$-regular, has diameter $D$,
    and is vertex-transitive. Suppose the Markov chain of simulated annealing is held constant at
    inverse temperature $\beta_t = \ln t / a$, where $a \geq J_{\max} := \max_{x \in S} J(x)$. Then
    \[
        \tau(\epsilon) \leq 128d^2D^4 t^2\ln\left(
            \frac{|S|t}{\epsilon} 
        \right).
    \]
\end{restatable}

This bound on the mixing time will allow us to control the non-equilibrium distribution of simulated annealing.
However, the quantity we are truly interested in is the expected value $\expect(J)$. Since the
logarithmic cooling schedule decreases the temperature very gradually, it seems likely
that this expected value never deviates far from its equilibrium value at a given temperature,
provided that the chain is allowed to mix sufficiently before the temperature is changed.
The following lemma formalizes this notion.
     
\begin{restatable}{lemma}{momentsnonequilibrium}\label{lemma:moments_nonequilibrium}
    Let $\nu_r = \nu_0 P_{t}^r$, where $P_{t}$ is the transition matrix
    of simulated annealing at temperature $a / \ln t$, with
    $a \geq n$ and $t = 2^{o(n)}$.
    If $\|\nu_0 - \pi_t\|_\mathrm{TV} \leq \frac{1}{|S|}$, then 
    \[
        \left|\expect_{\nu_r}( J \given W) - \expect_{\pi_{t+1}} (J \given W)\right|
            = O\left(\frac{n}{\sqrt{t}}\right).
    \]
    Moreover, the lower bound of \Cref{corollary:cost} holds also for
    $\expect_\mu(\expect_{\nu_r}(J \given W))$.
\end{restatable}

We are then led to the main result of this section.

\nonequilibriumcost

\begin{proof}
    
    Note first that the 2-opt state graph has a diameter of at most $n-1$ \cite{michiels_properties_2007} and
    is $\Theta(n^2)$-regular. \Cref{lemma:mixing_time} then gives us an upper bound for the mixing
    time of simulated annealing. 
    
    Let the mixing time at epoch    $t$ (and thus at a temperature of
    $a / \ln t$) be given by
    $\tau_t(\epsilon)$. We have
    \[
        \tau_t(\epsilon) \leq O(n^8 t^2) \ln\left(
            \frac{(n-1)!t}{2\epsilon} 
        \right) 
        \leq O(n^8 t^2) \left(
            n \log n + \log t + \log \frac{1}{\epsilon}
        \right).
    \]
    As $t = 2^{o(n)}$, we have $\log t = o(n)$. We can then choose
    $\epsilon = \frac{1}{|S|}$, so that $\log \frac{1}{\epsilon} = O(n \log n)$. This yields
    \begin{align*}
        \tau_t\left(\frac{1}{|S|}\right) = O\left(n^9t^2\log n\right).
    \end{align*}
    
    Hence, keeping the Markov chain at a constant temperature for $\Omega(n^9 t^2 \log n)$ steps
    suffices to bring the distribution within $\frac{1}{|S|}$ of the stationary distribution
    at the given temperature. We are then in the setting of \Cref{lemma:moments_nonequilibrium},
    which implies the result.
\end{proof}

\section{Lower Tail Bound}\label{sec:tailbound}

\Cref{thm:lb_nonequilibrium} tells us that the simulated annealing version under consideration here
will, in expectation, yield a tour of size linear in $n$ for $t = 2^{o(n)}$.
Taken together with the fact that the optimal tour is $\Theta(1)$ with high
probability \cite{engels_average-case_2009,frieze_random_2004},
this suggests that the algorithm stays
far from the global optimum for $t = 2^{o(n)}$. To strengthen this statement, we would like
to obtain a lower tail bound on the tour length resulting from executing SA for $t$ iterations.

Let $J(t)$ denote the length of a tour obtained on running simulated annealing for $t$ iterations
on a random instance. Informally speaking, if we could show that
$\prob\left(J(t) \leq (1 - \epsilon)\expect(t)\right)$ decreases exponentially quickly in $n$
for fixed $t$,
this would imply that the optimum is encountered
with probability $o_n(1)$ for $t = 2^{o(n)}$.

\subsection{Quenching versus Annealing}\label{sec:quenched_vs_annealed}

Recall the definitions
of the quenched and annealed distributions
of \cref{eq:quenched,eq:annealed}.
Given the physical interpretation of $\pi^a_\beta$ and $\pi^q_\beta$, one
might expect $\pi^a_\beta$ to tend towards lower values of $J$. Indeed,
we have previously shown this to hold in expectation (\Cref{lemma:j_lower_bound}). The question then
arises whether we can show a stronger relationship between these distributions.
In particular, a lower tail bound on $J$ would be straightforward to derive
if we could show that $\pi^q_\beta$ stochastically dominates $\pi^a_\beta$.
This of course hinges on the form of $\rho$. If $\mu$ is the uniform
distribution on $[0,1]^m$ as in the previous sections,
then $\rho$ is exactly the Irwin-Hall distribution;
in this case, the problem becomes analytically tractable.

We conjecture the following:
\stochdomconj
It is straightforward to show that this conjecture holds for $n = 3$. Admittedly, this
gets rid of most of the difficulty, since $\pi^q_\beta$ simply reduces to $\rho$. However,
$\pi^a_\beta$ is somewhat nontrivial, so we include the proof in the appendix.
\begin{restatable}{lemma}{stochdom}
    \Cref{conjecture:stochdom} holds for $n = 3$. 
\end{restatable}

Unfortunately, the proof does not easily extend to $n > 3$. Thus, we devote
the remainder of this section to accumulating numerical evidence that
\Cref{conjecture:stochdom} holds for larger $n$ as well.

As a final remark, we note that a proof of \Cref{conjecture:stochdom}
would immediately imply the $k = 1$ case of \Cref{lemma:j_lower_bound}.

\subsection{Numerical Results}

In order to gather data to support \Cref{conjecture:stochdom}, we must sample from
$\pi^q_\beta$ and $\pi^a_\beta$ at some fixed
value of $\beta$ and $n$, and compare the empirical CDF obtained from
these data. An expedient way to do this is by the
Metropolis-Hastings (MH) algorithm \cite{metropolis_equation_1953}.
This is equivalent to running
Algorithm \ref{sa_def} with a constant cooling schedule.

While we have thus far considered edge weights distributed
according to $U[0,1]^m$, we take a slightly different approach here,
which is more convenient for these numerical experiments. We define
$\mu$ as the uniform distribution over $\mathcal{W}^m$, where
$\mathcal{W} = \{i/N\}_{i=0}^N$,
with $N \in \mathbb{N}$. Note that $\mu$ converges to
$U[0,1]^m$ as $N \to \infty$.

For $\pi^q_\beta$, we generate
samples by first generating edge weights from $\mu$. We then run MH
for enough iterations such that the chain is sufficiently mixed.
At the end of this process, we output a single tour length, thereby obtaining one
sample from $\pi^q_\beta$. We then repeat this process for as many realizations
of the edge weights as we deem necessary.

For $\pi^a_\beta$, we similarly generate random edge weights from $\mu$, and sample using MH.
At every step of the algorithm, we choose an edge $e \in E$ uniformly at random,
and either increase or decrease its weight by $1/N$; this way, we allow
the edge weight vector to move within $\mathcal{W}^m$. Subsequently,
a 2-opt step is applied, and the usual MH acceptance step is performed. We
allow this Markov chain to mix sufficiently, and output a sample
once every $B$ iterations, where $B$ is an integer large enough such that the
samples are approximately independent.

Since we assume throughout that
$t = 2^{o(n)}$, and we are interested in the cooling schedule
of \Cref{eq:logcooling} with $a \geq n$, we have
$\beta = o_n(1)$. For this reason, we choose a constant value of $\beta$
here.
Performing these experiments for $n = 100$ and $n = 500$, with $N = 50$, yields
\Cref{fig:n100_500cdf}. For this experiment, we set $B = 10^4$,
collecting $10^4$ samples in total from $\pi_\beta^a$ and
the same number from $\pi_\beta^q$.
\Cref{fig:n100_500cdf} is suggestive of a stochastic dominance relationship
between $\pi^a_\beta$ and $\pi^q_\beta$. 

\begin{figure}[t]
    \centering
    \includegraphics[width=\textwidth]{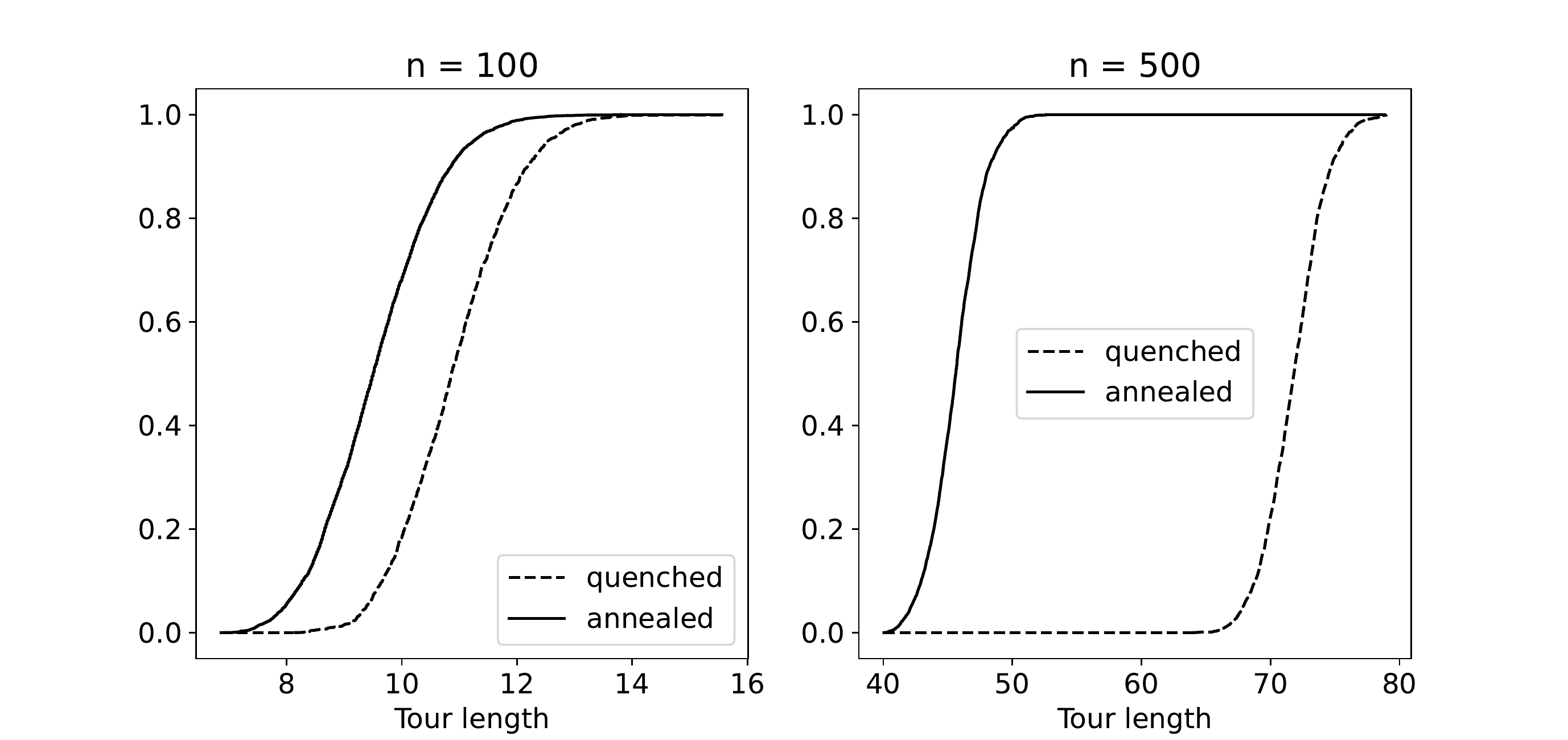}
    \caption{Empirical cumulative distribution functions of $\pi_\beta^q$ and
    $\pi_\beta^a$, for $\beta = 10$ at $n = 100$ and $n = 500$ with $N = 50$.}
    \label{fig:n100_500cdf}
\end{figure}

\section{Discussion \& Open Problems}

Our results highlight a connection between the physics of disordered systems, and
the probabilistic analysis of simulated annealing. To be precise,
in order to analyze the behavior of simulated annealing on random instances,
it is natural to define a Gibbs measure which depends on random parameters.
This is the exact same situation which occurs in disordered systems.
Hence, we believe that advances in the rigorous understanding of the statistical
mechanics of disordered systems will be advantageous to algorithms research.

To be more specific, our results hint that simulated annealing with
the logarithmic cooling schedule is not an efficient
solution method for TSP, even when compared with exact algorithms.
For simulated annealing, our results
hold for $t = 2^{o(n)}$, while
the well-known Bellman-Held-Karp algorithm
\cite{bellman_dynamic_1962,held_dynamic_1962}
has time complexity $O^*(2^n)$. 
Thus,
based on our results, we believe that
SA with logarithmic cooling cannot improve on this complexity, even
to a constant approximation ratio, although we lack a formal proof.
The logarithmic cooling schedule was previously
show to be necessary and sufficient for SA to 
asymptotically reach the global optimum \cite{hajek_cooling_1988}.
This paper is a first step towards a proof that this nice mathematical
property should be sacrificed in practice, a fact already well established
in practice \cite[chapter 8]{aarts_local_2003}.

The connection between
disordered systems and NP-hard problems is not novel, and has been
analyzed both from a physical
(e.g.\ \cite{baskaran_statistical_1986,fu_application_1986,kirkpatrick_critical_1994})
and mathematical
\cite{talagrand_high_2001,talagrand_spin_2003} perspective. However, much
of the research in physics has focused on the limits as
$\beta \to \infty$ and $n \to \infty$, referred to in that field
as the zero-temperature and thermodynamic limits. Research in mathematics
has focused mostly on spin glass systems, $k$-SAT, the
random assignment problem, and perceptrons.
From that perspective, the main novelty of our results is as a rigorous
statement on random TSP instances in the framework of disordered systems,
for finite $n$ and at finite temperature, which is of more
practical interest to the algorithms community than the aforementioned
limiting cases.

To conclude, we list some possible directions for future research on
the present topic.

\paragraph*{Lower tail bound out of equilibrium.}
We have already highlighted one open problem (\Cref{conjecture:stochdom}),
which would strengthen our results.
Even so, this conjecture is only a statement
on the equilibrium behavior of simulated annealing. If it can be shown to
hold, an obvious extension would be to the out-of-equilibrium situation,
which is more directly applicable to simulated annealing. This seems to
be a challenging problem, as it involves the time-dependent distribution
of an inhomogeneous Markov chain. Any rigorous results in this direction
may also yield insight into nonequilibrium statistical mechanics.

\paragraph*{Probabilistic models.}
Another clear open problem would be to analyze simulated annealing for
different probabilistic TSP models, especially models with dependence
between the edge weights. Physicists have previously considered
the Euclidean TSP, but rigorous results do not appear
to be known. The addition of dependence between edge weights would
make our proof of \Cref{lemma:partition_function} invalid, so this extension
would require more advanced proof techniques than the straightforward
analysis we have employed.

\paragraph*{Relaxed mixing requirements.}
In \Cref{sec:mixing} we have made some effort to prove statements
on simulated annealing outside of equilibrium. The results we have found are however
quite restrictive, requiring the Markov chain to be close to
equilibrium at all times, and moreover setting
strong limits on the cooling schedule. Proving an equivalent
to \Cref{lemma:moments_nonequilibrium} without the assumption
that $a \geq n$ is another challenge we have not addressed.

\paragraph*{Upper bound on the average case runtime.} While we have
in this work set out towards an upper bound on the average case runtime,
the reverse problem is also of interest. Since the known
worst case results are no improvement over
simply picking a tour uniformly at random
\cite{aarts_simulated_1989,hajek_cooling_1988,rajasekaran_convergence_1990},
we expect that it should be possible to show a better upper bound
in the average case.

\paragraph*{Application to different problems.}
Finally, some of our results are quite agnostic with respect to the
problem one analyzes. Replacing \Cref{lemma:partition_function} with
the partition function of any other problem would yield results applied
to that problem. It may thus be interesting to perform this
substitution for different problems, and identify problems
for which the negative results suggested by \Cref{corollary:cost,thm:lb_nonequilibrium}
are not applicable. These are then problems for which
the runtime lower bound we set out to prove may not hold.
Thus it may be possible to apply
simulated annealing with the logarithmic cooling schedule efficiently
to these problems.

\printbibliography

\appendix

\section{Proofs Omitted in Section \ref{sec:equilibrium}}

\cumulants*

\begin{proof}
    We start by noting that $Z(\beta \given W)$ is related to the moment generating function
    $M_J(s \given \cdot)$ of the random variable $J$. To wit,
    \[
        M_J(s \given W) := \expect_{\pi_\beta}\left(
            e^{sJ} 
        \right) = \sum_{x \in S} e^{sJ(x \given W)} \pi_\beta(x \given W)
        = \frac{\sum_{x \in S} e^{(s -\beta) J(x \given W)}}{\sum_{y \in S} e^{-\beta J(y \given W)}}
        = \frac{Z(\beta - s \given W)}{Z(\beta \given W)}.
    \]
    Since the cumulant generating function $K_J(s \given W)$ is defined as $K_J(s \given W) := \ln M_J(s \given W)$, we have
    \[
        K_J(s \given W) = \ln Z(\beta - s \given W) - \ln Z(\beta \given W).
    \]
    Then we have
    \[
        \kappa_p(\beta \given W) = \left.\frac{\partial^p}{\partial s^p}K_J(s \given W)\right|_{s=0} = (-1)^p \frac{\partial^p}{\partial \beta^p} \ln Z(\beta \given W).
    \]
\end{proof}

\jlowerbound*

\begin{proof}
    For the first inequality, suppose that is false, and integrate the resulting inequality to obtain
    \begin{align*}
        \expect(\ln Z(\beta\given W)) - \expect(\ln Z(0\given W))
            > \ln \expect(Z(\beta \given W)) - \ln \expect(Z(0\given W)). 
    \end{align*}
    Note that $Z(0\given W) = |S|$, independent of the realization of $W$. Hence,
    the corresponding terms on both sides cancel, and we are left with a contradiction
    to Jensen's inequality.
    
    Now for the second inequality, we proceed similarly, integrating its negation to find
    \begin{align*}
        \frac{\mathrm{d}}{\mathrm{d}\beta} \expect(\ln Z(\beta\given W))
            &- \left.\frac{\mathrm{d}}{\mathrm{d}\beta}\expect(\ln Z(\beta \given W))\right|_{\beta = 0} \\
            &> \frac{\mathrm{d}}{\mathrm{d}\beta} \ln \expect(Z(\beta \given W)) 
                - \left.\frac{\mathrm{d}}{\mathrm{d}\beta} \ln \expect(Z(\beta \given W))
                \right|_{\beta = 0}.
    \end{align*}
    We show that the second terms on both sides are again equal. Indeed:
    \begin{align*}
        -\frac{\mathrm{d}}{\mathrm{d}\beta} \expect(\ln (Z(\beta \given W)))
            = \expect\left(
                \sum_{x \in S} \frac{J(x\given W) e^{-\beta J(x \given W)}}{Z(\beta \given W)}
            \right),
    \end{align*}
    and so
    \begin{align*}
        -\left.\frac{\mathrm{d}}{\mathrm{d}\beta} \expect(\ln (Z(\beta \given W)))\right|_{\beta = 0}
            = \expect\left(
                \frac{1}{|S|}\sum_{x \in S}J(x\given W)
            \right) = \frac{1}{|S|}\sum_{x \in S}\expect(J(x\given W)).
    \end{align*}
    Meanwhile,
    \begin{align*}
        -\frac{\mathrm{d}}{\mathrm{d}\beta} \ln \expect(Z(\beta\given W))
            = \frac{1}{\expect(Z(\beta \given W))} \expect\left(
                \sum_{x \in S} J(x\given W)e^{-\beta J(x\given W)}
            \right).
    \end{align*}
    Evaluating this at $\beta = 0$ yields the same expression once more.
    Therefore, the corresponding terms
    once more cancel, leaving us with a contradiction to the first part of the lemma.
\end{proof}

\partitionfunction*

\begin{proof}
    
    By definition, we have $Z(\beta \given W) = \sum_{x \in S} e^{-\beta J(x \given W)}$. Recall that $J(x \given W)$ is a sum of
    edge weights, namely
    \[
        J(x \given W) = \sum_{e \in E(x)} W(e),
    \]
    where $E(x)$ denotes the set of edges present in tour $x$.
    Hence, the exponential in each term of the partition function can be factored into a product
    of exponential factors:
    \[
        e^{-\beta J(x \given W)} =  \prod_{e \in E(x)} e^{-\beta W(e)}.
    \]
    Using linearity of expectation, we then obtain 
    \[
        \expect(Z(\beta \given W)) = \sum_{x \in S} \prod_{e \in E(x)} \expect\left(e^{-\beta W(e)}\right),
    \]
    which holds since the edge weights are independent random variables.
    
    These expectations are easily evaluated:
    \[
        \expect\left(e^{-\beta W(e)}\right) = \int_0^1 e^{-\beta y} \mathrm{d}y = \frac{1 - e^{-\beta}}{\beta},
    \]
    since $W(e) \sim U[0, 1]$.
    Thus,
    \[
        \expect(Z(\beta \given W)) = \sum_{x \in S} \left(\frac{1 - e^{-\beta}}{\beta}\right)^n. 
    \]
    The fact that $|S| = \frac{1}{2}(n-1)!$ completes the proof.
\end{proof}

\section{Proofs Omitted in Section \ref{sec:mixing}}

For the following few proofs, recall the definitions stated at the start
of \Cref{sec:mixing}.
The mixing time of a Markov chain is related to the spectral properties of its 
transition matrix. 
To be precise, for a transition matrix $P$ with stationary distribution
$\pi$, let
\[
    \lambda_\star = \max \{ |\lambda| : \lambda \in \Lambda(P),\, \lambda \neq 1 \}
\]
be the largest (in absolute value) eigenvalue of $P$ not associated with $\pi$. Likewise, let
\[
    \lambda_1 = \max \{ \lambda : \lambda \in \Lambda(P),\, \lambda \neq 1 \}.
\]
Then the relaxation time of the chain is defined as
\[
    \tau_\mathrm{rel} := \frac{1}{1 - \lambda_\star}.
\]
The following holds:

\begin{lemma}[\mbox{Levin \& Peres~\cite[Thm 12.4]{aldous_markov_2019}}]\label{thm:mixing_time}
    Let $P$ be the transition matrix of a reversible, irreducible Markov chain with
    state space $S$, and let $\pi_{\min} := \min_{x \in S} \pi(x)$. Then
    \[
        \tau(\epsilon) \leq \tau_\mathrm{rel}\ln\left(
            \frac{1}{\epsilon \pi_{\min}}
        \right).
    \] 
\end{lemma}

We define the \emph{spectral gap} of $P$ as $\gamma = 1 - \lambda_1$, and distinguish it from the
the \emph{absolute spectral gap}, $\gamma_\star = 1 - \lambda_\star$.
This allows us to use a result from
Sinclair \& Jerrum:

\begin{lemma}[\mbox{Sinclair \& Jerrum~\cite[Prop. 3.2]{sinclair_approximate_1989}}]\label{lemma:modified_chain}
    Let P be the transition matrix of an ergodic time-reversible Markov chain, and 
    $1 = \lambda_0 > \lambda_1 > \ldots > \lambda_{N-1} > -1$ its eigenvalues.
    Then the modified chain with transition matrix $P' = \frac{1}{2}(I_N + P)$ is also ergodic
    and time-reversible with the same stationary distribution, and its eigenvalues
    $\{\lambda_i'\}$, similarly ordered, satisfy 
    $\lambda'_{N-1} > 0$ and $\lambda'_{\max{}} = \lambda_1' = \frac{1}{2}(1 + \lambda_1)$.
\end{lemma}

The modified chain in this proposition is easily constructed from the original
chain. At each iteration, one first flips a coin. On heads, the chain stays in the current state.
On tails, the usual Metropolis trial is conducted. Such a chain is called a lazy
Markov chain \cite{aldous_markov_2019}.

Given any realization of this lazy chain, we can obtain a realization of the
original Markov chain by removing certain elements from the generated sequence of states.
Since the chains share their stationary distribution, this implies that the mixing time of
the original chain is at most that of the lazy chain. In the following few lemmas,
we consider only the lazy chain.
The reason for this switch is that for the lazy chain, we have
$\gamma_\star = \gamma$,
which is easier to bound.

Finally, we quote two more results: one relating the spectral gap of a transition matrix to the structure of the
underlying state graph, and one relating the spectral gap to the \emph{bottleneck ratio} of the chain.
The bottleneck ratio of a Markov chain with transition matrix $P$ and stationary distribution $\pi$ is
defined as
\[
    \Phi_\star := \min_{\substack{S \subset V \\ \pi(S) \leq 1/2}}
        \frac{\sum_{x \in S} \sum_{y \in V \setminus S} \pi(x) P(x, y)}{\pi(S)},
\]
where $\pi(S) = \sum_{x \in S} \pi(x)$.

\begin{lemma}[\mbox{Levin \& Peres~\cite[Thm. 13.26]{aldous_markov_2019}}]\label{lemma:diameter_bound}
    Let $G$ be a transitive graph with vertex degree $d$ and diameter $D$. For the simple random walk
    on $G$,
        \[
            \frac{1}{\gamma} \leq 2dD^2. 
        \]
\end{lemma}

\begin{lemma}[\mbox{Sinclair \& Jerrum (1989), Lawler \& Sokal (1988)~\cite[Thm. 13.10]{aldous_markov_2019}}]\label{lemma:spectral_gap}
    Let $\lambda_1$ be the second largest eigenvalue of a reversible transition matrix $P$, and
    let $\gamma = 1 - \lambda_1$. Then
    \[
        \frac{\Phi_\star^2}{2} \leq \gamma \leq 2\Phi_\star. 
    \]
\end{lemma}

\mixingtime*

\begin{proof}

    Let $P_t$ be the transition matrix of the Markov chain of lazy simulated annealing at
    inverse temperature $\beta_t = \ln t / a$.
    Let $\Phi(t)$ denote the bottleneck ratio of the Markov chain at inverse temperature
    $\beta_t$, and let $\Phi_0 = \lim_{t \to 0} \Phi(t)$. 

    Following a proof by Nolte and Schrader \cite[Theorem 1]{nolte_note_2000}, we have
    \[
        \Phi(t) \geq \Phi_0 t^{-J_{\max}/a} \geq \frac{\Phi_0}{t}.
    \]
    From the definition of $\Phi_\star$, one easily sees that the bottleneck
    ratio of a lazy Markov chain is identical to that of the original chain,
    multiplied by a factor of $1/2$. After all, the off-diagonal elements of the lazy
    transition matrix are simply those of the original matrix multiplied by $1/2$,
    and only these elements appear in the definition of $\Phi_\star$. Hence,
    letting $\Phi$ denote the bottleneck ratio of the simple random walk,
    we find $\Phi(t) \geq \Phi_0/t = \Phi/(2t)$.
    
    \Cref{lemma:spectral_gap} tells us that $\Phi \geq \gamma/2$, where $\gamma$ denotes the spectral
    gap of $P_0$. Now by \Cref{lemma:diameter_bound}, we have
    \[
        \gamma \geq \frac{1}{2dD^2}. 
    \]
    Putting all this together, we arrive at
    \[
        \Phi(t) \geq \frac{\Phi}{2t} \geq \frac{\gamma}{4t} \geq \frac{1}{8dD^2 t}. 
    \]
    Then by \Cref{lemma:spectral_gap},
    \[
        \lambda_\star \leq 1 - \frac{\Phi(t)^2}{2} \leq 1 - \frac{1}{128d^2 D^4 t^2}.
    \] 
    Now by \Cref{thm:mixing_time},
    \[
        \tau(\epsilon) \leq  \frac{1}{1 - \lambda} \ln\left(
            \frac{1}{\epsilon \min_{x \in S} \pi_t(x)} 
        \right) \leq 128d^2 D^4 t^2 \ln\left(
            \frac{1}{\epsilon \min_{x \in S} \pi_t(x)} 
        \right).
    \]
    The stationary distribution has a straightforward lower bound:
    \[
        \pi_t(x) = \frac{1}{Z(\beta_t \given w)}e^{-\beta_t J(x)} \geq \frac{1}{Z(\beta_t \given w)} e^{-\beta_t J_{\max}}
            \geq \frac{1}{|S|}e^{-\beta_t J_{\max}},
    \]
    which follows from the fact that $Z(\beta_t \given w) \leq |S|$. Hence,
    \[
        \tau(\epsilon) \leq 128d^2D^4 t^2 \ln\left(
            \frac{|S|e^{\beta_t J_{\max}}}{\epsilon}
        \right)
        \leq 128d^2D^4 t^2\ln\left(
            \frac{|S|t}{\epsilon} 
        \right).
    \]
    Since the mixing time of a lazy Markov chain is an upper bound to the
    mixing time of the original chain, we are done.
\end{proof}

\momentsnonequilibrium*

\begin{proof}
    For notational convenience, denote $\expect_{\nu}(\cdot \given W)$ by
    $\langle \cdot \rangle_{\nu}$ for a conditional probability measure $\nu(\cdot \given W)$ defined
    on $S$. Similarly, we will suppress the dependence of $\nu_r$, $\pi_{t}$ and $J$ on $W$
    as it plays no role in this proof.
    We begin by bounding the difference in expectation of $J$ taken over $\pi_{t+1}$ and $\nu_r$,
    \begin{align*}
        \left|\langle J \rangle_{\nu_r} - \langle J \rangle_{\pi_{t+1}}\right|
            &= \left|\sum_{x \in S} J(x)(\nu_r(x) - \pi_{t+1}(x))\right| \\
            &= \left|\sum_{x \in S} J(x)\sqrt{\pi_{t+1}(x)} 
                \frac{\nu_r(x) - \pi_{t+1}(x)}{\sqrt{\pi_{t+1}(x)}}\right|,
    \end{align*}
    which is permissible since $\pi_{t+1}(x) > 0$ for all $x$ when $t < \infty$.
    Let $\Pi_t = \diag(\pi_t)$, and consider $J \in \mathbb{R}^{|S|}$ as a vector on $S$.
    Then the above can be cast as
    \begin{align*}
        \left|\langle J \rangle_{\nu_r} - \langle J \rangle_{\pi_{t+1}}\right|
            =\left|\left\langle\Pi_{t+1}^{1/2}J, \Pi_{t+1}^{-1/2}(\nu_r - \pi_{t+1})\right\rangle\right|
            \leq \|\Pi_{t+1}^{1/2}J\|_2\left\|\Pi_{t+1}^{-1/2}(\nu_r - \pi_{t+1})\right\|_2,
    \end{align*}
    where $\langle \cdot, \cdot \rangle$ denotes the standard Euclidean inner product.
    Now, we define the weighted norm
    \[
        \|u\|_{t} := \sqrt{\sum_{x \in S}\frac{u(x)^2}{\pi_t(x)}} = \|\Pi_t^{-1/2}u\|_2.
    \]
    Furthermore, notice that $\|\Pi_{t+1}^{1/2}J\|_2 = \sqrt{\langle J^2 \rangle_{\pi_{t+1}}}$. Hence,
    \[
        \left|\langle J \rangle_{\nu_r} - \langle J \rangle_{\pi_{t+1}}\right|
            \leq \sqrt{\langle J^2 \rangle_{\pi_{t+1}}} \|\nu_r - \pi_{t+1}\|_{t+1}
            \leq n \|\nu_r - \pi_{t+1}\|_{t+1},
    \]
    since $J(x) \leq n$ for all $x \in S$.
    It then remains to show a suitable bound for $\|\nu_r - \pi_{t+1}\|_{t+1}$. 
    To that end,  the following fact is useful.
    
    \begin{claim}
        For $A \in \mathbb{R}^{|S|\times |S|}$, let $\|A\|_t$ denote the matrix norm induced by
        $\| \cdot \|_{t}$. We have $\|P_t\|_t = 1$.
    \end{claim}
    
    \begin{claimproof}
        By the definition of $\|\cdot\|_t$, it follows that $\|A\|_t = \|\Pi_t^{1/2} A \Pi_t^{-1/2}\|_2$.
        Define $Q_t = \Pi_t^{1/2}P_t\Pi_t^{-1/2}$. One can show that
        under detailed balance for $P_t$, the matrix $Q_t$ is symmetric. As a result,
        $\|Q_t\|_2 = \max_{\lambda \in \Lambda(Q_t)} |\lambda|$, where $\Lambda(A)$ denotes the spectrum
        of the operator $A$. Since $Q_t$ and $P_t$ are similar, $\Lambda(Q_t) = \Lambda(P_t)$, and since
        $P_t$ is stochastic, its largest eigenvalue in absolute value is $1$. Putting all this together leads
        us to conclude $\|P_t\|_t = 1$.
    \end{claimproof}
    
    Armed with this fact and the sub-multiplicativity of the matrix norm, we find
    \[
        \|\nu_r - \pi_{t+1}\|_{t+1} = \|(\nu_0 - \pi_{t+1})P_{t+1}^r\|_{t+1}
            \leq \|\nu_0 - \pi_{t+1}\|_{t+1}\|P_{t+1}\|_{t+1}^r
            = \|\nu_0 - \pi_{t+1}\|_{t+1}.
    \]
    Applying the triangle inequality,
    \[
        \|\nu_0 - \pi_{t+1}\|_{t+1} \leq \|\pi_{t+1} - \pi_t\|_{t+1} + \|\nu_0 - \pi_t\|_{t+1}.
    \]
    For the first term, by a result of Nolte \& Shrader \cite{nolte_note_2000}, we have
    \[
        \|\pi_{t+1} - \pi_t\|_{t+1} \leq \sqrt{\frac{1}{t}}. 
    \]
    For the second term,
    \[
        \|\nu_0 - \pi_t\|_{t+1} \leq \left\|\Pi_{t+1}^{-1/2}\right\|_2 \|\nu_0 - \pi_t\|_2
            \leq \frac{\|\nu_0 - \pi_t\|_2}{\sqrt{\min_{x \in S} \pi_{t+1}(x)}}.
    \]
    Using the same lower bound on $\pi_{t+1}$ as in the proof of \Cref{lemma:mixing_time}:
    \[
        \min_{x \in S} \pi_{t+1}(x)
            \geq \frac{1}{|S|(t+1)},
    \]
    as $Z(\beta \given W) \leq |S|$ and $a \geq n \geq J(x)$ for all $x \in S$. Hence,
    \begin{align*}
        \|\nu_0 - \pi_t\|_{t+1} &\leq \sqrt{|S|(t+1)}\|\nu_0 - \pi_t\|_2
            \leq \sqrt{|S|(t+1)}\|\nu_0 - \pi_t\|_1 \\
            &= 2\sqrt{|S|(t+1)}\|\nu_0 - \pi_t\|_\mathrm{TV}.
    \end{align*}
    Since by assumption $\|\nu_0 - \pi_t\|_\mathrm{TV} \leq \frac{1}{|S|}$,
    \[
        \|\nu_0 - \pi_t\|_{t+1} \leq
            2\sqrt{\frac{t+1}{|S|}}.
    \]
    Now recall that $t = 2^{o(n)}$. Therefore, since
    $|S| = (n-1)!/2$, we have $(t+1)/|S| = 2^{o(n)}/|S| \leq 2^{-O(n)} =
    O(1/t)$.
    We then conclude that
    \[
        \|\nu_0 - \pi_t\|_{t+1} =
            O\left(\sqrt{\frac{1}{t}}\right).
    \]
    
    Putting these two bounds together, we have
    \[
        \|\nu_0 - \pi_{t+1}\|_{t+1} = O\left(\sqrt{\frac{1}{t}}\right),
    \]
    which proves the first part of the lemma.
    The second part of the lemma follows directly, since we have
    \[
        \expect_{\nu_r}(J \given W) \geq \expect_{\pi_{t+1}}( J \given W)
            - O\left(\frac{n}{\sqrt{t}}\right).
    \]
    Using the lower bound from \Cref{corollary:cost}, the result follows.
\end{proof}

\section{Proofs Omitted in Section \ref{sec:tailbound}}

\stochdom*

\begin{proof}
   We have
    \[
        \prob(J_a \leq j) = \frac{\int_0^j \rho(j') e^{-\beta j'}\mathrm{d}j'}{Z_a(\beta)/|S|},
    \]
    where $Z_a(\beta) = \expect_\mu(Z(\beta \given W))$,
    and
    \[
        \prob(J_q \leq j) = \int_{[0,1]^m}
            \mu(w)\frac{\int_0^j \rho(j'\given W=w) e^{-\beta j'}\mathrm{d}j'}{\int_0^n \rho(j'\given W=w)e^{-\beta j'}\mathrm{d}j'}
            \mathrm{d}w.
    \]
    By \Cref{lemma:partition_function},
    \[
        Z_a(\beta) = |S| \left(\frac{1 - e^{-\beta}}{\beta}\right)^n = \left(\frac{1 - e^{-\beta}}{\beta}\right)^n,
    \]
    since for $n = 3$, we have $|S| = 1$. Moreover, since there
    is only one tour, $\prob(J_q \leq j)$ reduces to
    $F_\rho(j)$, the CDF of $\rho$.
    For $J_a$, we find
    \begin{align*}
        \prob(J_a \leq j) &= \left(\frac{\beta}{1 - e^{-\beta}}\right)^n
            \int_0^j \rho(j')e^{-\beta j'} \mathrm{d}j' \\
                &= \left(\frac{\beta}{1 - e^{-\beta}}\right)^n
                    \expect_{\rho}(e^{-\beta J} \given J \leq j)F_\rho(j)\\
                    &\geq 
                        \left(\frac{\beta}{1 - e^{-\beta}}\right)^n
                    \expect_{\rho}(e^{-\beta J})F_\rho(j) \\
                    &= F_\rho(j) = \prob(J_q \leq j).
    \end{align*}
    Hence, $\prob(J_q \geq j) \geq \prob(J_a \geq j)$, and therefore
    $J_q \succeq J_a$ for $n = 3$.
\end{proof}

\end{document}